\documentclass[11pt]{article}

\usepackage[a4paper,margin=29mm]{geometry}
\usepackage[T1]{fontenc}
\usepackage[utf8]{inputenc}
\usepackage{lmodern}
\usepackage{microtype}
\usepackage{amsmath,amssymb,amsfonts,amsthm,mathtools,mathrsfs}
\usepackage{aliascnt}
\usepackage{bm}
\usepackage{enumitem}
\usepackage{booktabs}
\usepackage{array}
\usepackage{csquotes}
\usepackage[backend=biber,style=authoryear,natbib=true,maxcitenames=2,maxbibnames=99,uniquelist=false,uniquename=init,giveninits=true,doi=true,url=false,isbn=false]{biblatex}
\usepackage{xcolor}
\usepackage{hyperref}
\usepackage[nameinlink,noabbrev]{cleveref}

\crefname{theorem}{theorem}{theorems}
\Crefname{theorem}{Theorem}{Theorems}
\crefname{lemma}{lemma}{lemmas}
\Crefname{lemma}{Lemma}{Lemmas}
\crefname{definition}{definition}{definitions}
\Crefname{definition}{Definition}{Definitions}
\crefname{assumption}{assumption}{assumptions}
\Crefname{assumption}{Assumption}{Assumptions}
\crefname{proposition}{proposition}{propositions}
\Crefname{proposition}{Proposition}{Propositions}
\crefname{corollary}{corollary}{corollaries}
\Crefname{corollary}{Corollary}{Corollaries}
\crefname{remark}{remark}{remarks}
\Crefname{remark}{Remark}{Remarks}

\hypersetup{
  colorlinks=true,
  linkcolor=blue!45!black,
  citecolor=blue!45!black,
  urlcolor=blue!45!black,
  pdftitle={Spending Operators and Weak Power-Port Decompositions for Path-Dependent Entropic Lagrangians},
  pdfauthor={Huilong Ren},
  pdfsubject={Calibrated terminal spending, diffusion power-port decomposition, recovery of local species balance, and a single thermo-diffusion functional},
  pdfkeywords={terminal variation, cocycle calibration, weak Leibniz rule, chemical-potential weighting, local species balance, Cahn--Hilliard equation}
}

\setlist[itemize]{leftmargin=2em,itemsep=0.2em,topsep=0.3em}
\setlist[enumerate]{leftmargin=2.2em,itemsep=0.2em,topsep=0.3em}
\allowdisplaybreaks

\numberwithin{equation}{section}
\newtheorem{theorem}{Theorem}[section]
\newaliascnt{proposition}{theorem}
\newtheorem{proposition}[proposition]{Proposition}
\aliascntresetthe{proposition}
\newaliascnt{lemma}{theorem}

\aliascntresetthe{lemma}
\newaliascnt{corollary}{theorem}
\newtheorem{corollary}[corollary]{Corollary}
\aliascntresetthe{corollary}
\newaliascnt{assumption}{theorem}

\aliascntresetthe{assumption}
\theoremstyle{definition}
\newaliascnt{definition}{theorem}
\newtheorem{definition}[definition]{Definition}
\aliascntresetthe{definition}
\theoremstyle{remark}
\newaliascnt{remark}{theorem}

\aliascntresetthe{remark}

\newcommand{\R}{\mathbb{R}}
\newcommand{\dd}{\,\mathrm{d}}
\newcommand{\Div}{\operatorname{div}}
\newcommand{\grad}{\nabla}
\newcommand{\Ran}{\operatorname{Ran}}
\newcommand{\Ker}{\operatorname{Ker}}
\newcommand{\Int}{\operatorname{int}}
\newcommand{\Dcal}{\mathcal{D}}
\newcommand{\Ecal}{\mathcal{E}}

\newcommand{\Lcal}{\mathfrak{L}}
\newcommand{\dual}[2]{\left\langle #1,#2\right\rangle}
\newcommand{\norm}[1]{\left\lVert #1\right\rVert}
\newcommand{\abs}[1]{\left\lvert #1\right\rvert}
\newcommand{\set}[1]{\left\{#1\right\}}
\newcommand{\eps}{\varepsilon}

\title{\textbf{\Large\parbox{0.96\textwidth}{\centering
Spending Operators and Weak Power-Port Decompositions\\[-0.15em]
for Path-Dependent Entropic Lagrangians}}}
\author{Huilong Ren\\
College of Civil Engineering, Tongji University, Shanghai, China\\
\texttt{hlren@tongji.edu.cn}}
\date{}

\begin{document}
\maketitle

\begin{abstract}
History-dependent entropic variational formulations require a calibrated terminal differential for accumulated power and a spatial power split that remains meaningful for weak diffusion fields. Endpoint calibration and cocycle additivity determine a unique oriented spending increment, while independent local selector fields give its distributional channel form. The diffusion identity for a potential-weighted flux is established at $H(\Div)$ regularity and extended to the finite-energy class $H^1\times L^2$ through a distributional balance component. For regular diffusion models, the weighted species channel yields the local balance whenever persistent zero-potential states are dynamically isolated in the admissible state class. An independent multiplier extends the same balance to the natural weak space. These ingredients are assembled in one synchronized thermo-diffusion functional whose directional stationarity yields energetic and thermal conjugacy, species balance, flux closure, the entropy equation, and both terminal routing rules. A Cahn--Hilliard specialization verifies the construction at finite-energy regularity.
\end{abstract}

\noindent\textbf{Keywords:} terminal variation; cocycle calibration; weak Leibniz rule; chemical-potential weighting; local species balance; Cahn--Hilliard equation.

\noindent\textbf{MSC 2020:} 49J40; 49S05; 35K55; 35D30; 46N20; 74A15.

\section{Introduction}\label{sec:introduction}

Variational descriptions of irreversible evolution have developed from several choices of primary structure. Onsager's reciprocal relations and Biot's irreversible thermodynamics express dissipative evolution through force--rate pairings, while generalized standard materials and incremental constitutive principles encode material response by convex functions of rates or internal variables \citep{Onsager1931a,Onsager1931b,Biot1955,HalphenNguyen1975,OrtizStainier1999,Doi2011}. Generalized-gradient systems extend this idea to nonlinear and nonsmooth evolutions and provide a precise energy--dissipation structure once the state space, energy, and balance operator have been fixed \citep{MielkePeletierRenger2014,Mielke2023}. These theories supply constitutive contact and stability principles, but the accumulated process power is normally evaluated along the resulting trajectory rather than treated through an independently calibrated terminal variation.

Geometric and action-based formulations resolve complementary aspects of irreversible dynamics. Port-Hamiltonian and Stokes--Dirac structures organize bulk--boundary power exchange and have recently been developed in Hilbert-space settings that make the boundary dualities explicit \citep{vanDerSchaftMaschke2002,LeGorrecZwartMaschke2005,BrugnoliHaineMatignon2023}. Doubled-variable actions and constrained thermodynamic principles enlarge the variation space in order to accommodate nonconservative forces and entropy production \citep{Galley2013,FukagawaFujitani2012,GayBalmazYoshimura2017a,GayBalmazYoshimura2017b}. Recent extensions treat thermodynamic fluxes as independent variables, incorporate objective heat-flux rates, and formulate general dissipative fluid systems by differential forms \citep{GayBalmaz2025Extended,ManachPerenouGayBalmaz2026}. These developments demonstrate that irreversible processes can be embedded in a variational framework, while leaving open a separate question that is central here: the first-order routing of an already accumulated power into distinct terminal channels.

Diffusion introduces an additional analytical issue. The Cahn--Hilliard equation can be derived from microforce balances, energetic variational principles, and mixed rate formulations \citep{CahnHilliard1958,Gurtin1996,GurtinPolignoneVinals1996,MieheHildebrandBoger2014,LiuWu2019}. In weak solutions, however, the natural regularity is often only a chemical potential in $H^1(\Omega)$ and a flux in $L^2(\Omega;\mathbb R^d)$. The classical product formula for a potential-weighted flux must then be interpreted distributionally, and the normal trace must be retained when boundary power is present. The corresponding tools come from $H(\Div)$ theory, divergence-measure fields, and modern pairing formulas \citep{GiraultRaviart1986,ChenFrid1999,ChenComiTorres2019,ComiCrastaDeCiccoMalusa2024}. This weak-product issue is not merely notational: it determines whether the balance-side and entropy-side parts of the diffusion port are meaningful in the function spaces used by the evolution problem.

The Path-Dependent Entropic Lagrangian framework was introduced for irreversible thermomechanical systems and subsequently developed as a path-dependent variational formulation of irreversible thermodynamics \citep{ren2025path,ren2026variational}. Its defining operation is to route one accumulated process power through complementary terminal channels while deriving balance and thermodynamic conjugacy from the same variational construction. The present paper isolates the mathematical structures needed for that operation. First, synchronization of terminal variables determines only a diagonal set and does not fix the amplitude of the channel differential. Second, terminal selectors must be local test fields if their stationarity is to yield local equations. Third, the diffusion port requires a weak Leibniz decomposition compatible with finite-energy fields. Finally, the chemical-potential-weighted species channel must be connected to local conservation through the constitutive structure of zero-potential states, with an independent weak formulation available when the regularity needed for that argument is absent.

The main results follow this sequence. Endpoint calibration and the cocycle identity determine a unique oriented history extension and hence the spending rule, both for scalar terminal directions and for independent local selector fields. The potential-weighted diffusion flux is decomposed first at $H(\Div)$ regularity, including its boundary power, and then at the finite-energy level through a distributional balance component. A constitutive recovery theorem gives sufficient conditions under which the weighted species channel yields the local continuity equation, while a mixed multiplier provides the same balance directly in the natural dual space. These ingredients are then assembled in one synchronized thermo-diffusion functional. Its directional stationarity generates energetic and thermal conjugacy, local species balance, diffusion and heat-flux laws, the entropy equation, and the two terminal routing formulas. A Cahn--Hilliard specialization verifies the construction at finite-energy regularity.

The remainder of the paper is outlined as follows. \Cref{sec:setting} fixes the analytical setting and distinguishes vertical path variation from the history-locked terminal derivative. \Cref{sec:spending} establishes the calibrated spending rule and its localized selector fields. \Cref{sec:spatial} develops the Sobolev and finite-energy diffusion-port decompositions. \Cref{sec:balance} gives the constitutive recovery of local balance and the mixed weak completion. \Cref{sec:energy} records the energy balance and its relation to process power. \Cref{sec:single} constructs the synchronized thermo-diffusion functional, and \cref{sec:ch} gives the Cahn--Hilliard realization. \Cref{sec:conclusion} summarizes the conclusions; auxiliary multichannel and measure-valued extensions are collected in the appendices.

\section{Analytical setting}\label{sec:setting}

\subsection{Power and constitutive contact}

Let $X$ be a real reflexive Banach space with dual $X^*$. A path
\begin{equation}\label{eq:path_rate}
  z\in W^{1,1}(0,T;X),
  \qquad v(t):=\dot z(t)\in X\quad\text{for a.e. }t,
\end{equation}
has a conjugate effort $r(t)\in X^*$. The scalar power is
\begin{equation}\label{eq:power_definition}
  D(t):=\dual{r(t)}{v(t)}_{X^*,X}.
\end{equation}
Its sign is constitutive.

When constitutive contact is represented by a proper, convex, lower-semicontinuous rate-power function $\mathcal R:X\to[0,+\infty]$ with $\mathcal R(0)=0$, the standard Fenchel relation gives
\begin{equation}\label{eq:fenchel_contact}
  r\in\partial\mathcal R(v)
  \quad\Longleftrightarrow\quad
  \dual{r}{v}=\mathcal R(v)+\mathcal R^*(r)\ge0.
\end{equation}
This step identifies a nonnegative process power; it does not fix terminal orientation or the spatial channel decomposition \citep{Rockafellar1970,EkelandTemam1999,MielkePeletierRenger2014}. In the continuum examples below, the local constitutive functions are called \emph{flux power densities}.

\subsection{Vertical and terminal variations}

A vertical path variation is generated by an admissible curve
\begin{equation}\label{eq:vertical_variation_curve}
  z_\eps=z+\eps\eta,
  \qquad \eta\in T_z\mathcal X.
\end{equation}
For a smooth integrand $P=P(z,v,t)$, differentiation of a history with a moving endpoint gives
\begin{align}
 &\left.\frac{\dd}{\dd\eps}\right|_{\eps=0}
 \int_{t_0}^{t+\eps\xi}
 P\bigl(z_\eps(\tau),\dot z_\eps(\tau),\tau\bigr)\dd\tau
 \notag\\
 &\qquad=
 \int_{t_0}^{t}
 \left[D_zP(z,\dot z,\tau)[\eta]
       +D_vP(z,\dot z,\tau)[\dot\eta]\right]\dd\tau
 +P(z(t),\dot z(t),t)\xi,
 \label{eq:full_path_variation}
\end{align}
whenever differentiation under the integral is justified. The terminal operation used in this paper is the restriction $\eta=0$ in \eqref{eq:full_path_variation}; it excludes interior vertical variation by definition.

For $D\in L^1(t_0,T)$, set
\begin{equation}\label{eq:history_primitive}
  H_D(t):=\int_{t_0}^{t}D(\tau)\dd\tau.
\end{equation}
Then $H_D\in W^{1,1}(t_0,T)$ and $\dot H_D=D$ almost everywhere. If the power is generated along a realized path, write $D_z(t)=\dual{r_z(t)}{\dot z(t)}$ and $\mathscr H[z_{\le t}]=H_{D_z}(t)$.

\begin{definition}[History-locked terminal derivative]\label{def:history_locked_derivative}
Fix $t\in(t_0,T)$, a realized history, and its associated power $D_z\in L^1(t_0,T)$. The history-locked terminal derivative is
\begin{equation}\label{eq:history_locked_derivative}
  \partial_t^{\rm hl}\mathscr H[z_{\le t}]
  :=\lim_{h\to0,\;t+h\in(t_0,T)}
  \frac{H_{D_z}(t+h)-H_{D_z}(t)}{h},
\end{equation}
whenever the limit exists. The function $D_z$ is held fixed in this operation; no vertical perturbation of $z$ or of the constitutive law is taken inside the accumulated history.
\end{definition}

At every Lebesgue point of $D_z$,
\begin{equation}\label{eq:terminal_power_identity}
  \partial_t^{\rm hl}\mathscr H[z_{\le t}]=D_z(t).
\end{equation}
The distinction between vertical and terminal directions parallels the separation of path and time derivatives in nonanticipative path calculus, although no stochastic calculus is used here \citep{ContFournie2013}.

\section{Calibrated terminal spending}\label{sec:spending}

\subsection{Oriented history extension}

Let
\begin{equation}\label{eq:synchronization_diagonal}
  \Delta:=\set{(t,t):t\in(t_0,T)}.
\end{equation}
If a $C^1$ function $\widetilde\Gamma$ vanishes on $\Delta$, then
$D\widetilde\Gamma(t,t)=a(t)(1,-1)$ for some scalar $a(t)$, because the differential annihilates the tangent direction $(1,1)$. The diagonal condition determines the conormal direction but not its amplitude. Endpoint calibration is therefore an additional datum.

\begin{definition}[Calibrated oriented history extension]\label{def:calibrated_extension}
Let $D\in L^1(t_0,T)$. A map $\Gamma:[t_0,T]^2\to\R$ is a calibrated oriented history extension of $D$ if
\begin{enumerate}[label=\textup{(\roman*)}]
  \item $\Gamma(a,c)=\Gamma(a,b)+\Gamma(b,c)$ for all $a,b,c$;
  \item for every fixed $b$, the map $a\mapsto\Gamma(a,b)$ is absolutely continuous;
  \item at every Lebesgue point $t$ of $D$,
  \begin{equation}\label{eq:endpoint_calibration}
    \lim_{h\to0}\frac{\Gamma(t+h,t)}{h}=D(t),
  \end{equation}
  whenever $t+h\in[t_0,T]$.
\end{enumerate}
\end{definition}

\begin{theorem}[Uniqueness and spending rule]\label{thm:calibrated_spending}
For every $D\in L^1(t_0,T)$ there is exactly one calibrated oriented history extension:
\begin{equation}\label{eq:calibrated_extension_formula}
  \Gamma_D(t_s,t_c)
  =H_D(t_s)-H_D(t_c)
  =\int_{t_c}^{t_s}D(\tau)\dd\tau.
\end{equation}
At every Lebesgue point $t$ of $D$,
\begin{equation}\label{eq:spending_differential}
  D\Gamma_D(t,t)[\delta t_s,\delta t_c]
  =D(t)\delta t_s-D(t)\delta t_c.
\end{equation}
The common direction is null: $D\Gamma_D(t,t)[\xi,\xi]=0$.
\end{theorem}

\begin{proof}
Let $F(t)=\Gamma(t,t_0)$. The cocycle identity gives $\Gamma(a,b)=F(a)-F(b)$ and also $F(t_0)=0$. By absolute continuity and \eqref{eq:endpoint_calibration}, $F'=D$ almost everywhere, hence $F=H_D$. Differentiating \eqref{eq:calibrated_extension_formula} at $(t,t)$ gives \eqref{eq:spending_differential}.
\end{proof}

The argument separates the two ingredients of the construction. The cocycle identity reduces every admissible two-endpoint history to a difference of one primitive, whereas endpoint calibration fixes the derivative of that primitive. Hence the opposite signs in \eqref{eq:spending_differential} follow from orientation of the endpoints and are not prescribed independently after synchronization.

The directly usable notation is
\begin{equation}\label{eq:spending_rule_citable}
\boxed{
  \delta\int_{t_0}^{t}D(\tau)\dd\tau
  =D(t)\delta t_s-D(t)\delta t_c.
}
\end{equation}
Here $\delta$ denotes the diagonal differential of the calibrated extension, not two unrelated derivatives of one scalar upper limit. Equivalently, define
\begin{equation}\label{eq:spending_epsilon_family}
  H_\eps:=H_D(t)+
  \Gamma_D(t+\eps\delta t_s,t+\eps\delta t_c).
\end{equation}
Then $H_0=H_D(t)$ and $H_\eps'|_{\eps=0}$ is the right-hand side of \eqref{eq:spending_rule_citable}. This provides an ordinary one-parameter directional-derivative realization of the notation.

\subsection{Spending operator and trajectory-generated directions}

For a fixed rate $v\in X$, set
\begin{equation}\label{eq:compatibility_space}
  K_v:=\set{(v\xi,\xi):\xi\in\R}\subset X\times\R,
\end{equation}
and define
\begin{equation}\label{eq:spending_operator}
  \mathfrak S_v:X^*\to X^*\times\R,
  \qquad
  \mathfrak S_v(r):=\bigl(r,-\dual{r}{v}\bigr).
\end{equation}
A direct calculation gives
\begin{equation}\label{eq:spending_operator_range}
  K_v^\perp=\Ran\mathfrak S_v
  =\set{\bigl(r,-\dual{r}{v}\bigr):r\in X^*}.
\end{equation}
Consequently, if $D=\dual{r}{v}$ and $\eta=v\delta t_c$, then
\begin{equation}\label{eq:power_virtual_work}
  D\delta t_s-D\delta t_c
  =D\delta t_s-\dual{r}{\eta}.
\end{equation}
The relation $\eta=v\delta t_c$ selects one direction generated by the realized trajectory. It does not parametrize the full state-variation space. In vector mechanics, arbitrary virtual displacements remain independent directions. Moreover, the scalar $D$ does not determine $r$: every functional annihilating $v$ can be added to $r$ without changing $D$.

\subsection{Localized selector fields}

Let $I\subset\R$ be open, let $\Omega$ have finite measure, and let
$p\in C(I;L^1(\Omega))$ be a jointly measurable power density. For measurable endpoint fields define
\begin{equation}\label{eq:localized_history_increment}
  \Gamma_p[\tau_s,\tau_c]
  :=\int_\Omega\int_{\tau_c(x)}^{\tau_s(x)}
  p(x,\sigma)\dd\sigma\dd x.
\end{equation}

\begin{theorem}[Localized spending rule]\label{thm:localized_spending}
Fix $t\in I$ and let $\delta t_s,\delta t_c\in L^\infty(\Omega)$. For sufficiently small $\eps$, set
\begin{equation}\label{eq:localized_endpoint_fields}
  \tau_s^\eps=t+\eps\delta t_s,
  \qquad
  \tau_c^\eps=t+\eps\delta t_c.
\end{equation}
Then
\begin{equation}\label{eq:localized_spending_derivative}
\boxed{
  \left.\frac{\dd}{\dd\eps}\right|_{\eps=0}
  \Gamma_p[\tau_s^\eps,\tau_c^\eps]
  =\int_\Omega p(x,t)(\delta t_s-\delta t_c)\dd x.
}
\end{equation}
\end{theorem}

\begin{proof}
It suffices to treat one endpoint. For $h\in L^\infty(\Omega)$, let
\[
  F_\eps(h):=\int_\Omega\int_t^{t+\eps h(x)}p(x,\sigma)\dd\sigma\dd x,
  \qquad M:=\norm{h}_{L^\infty}.
\]
Fubini's theorem yields
\begin{align*}
  \left|\frac{F_\eps(h)}{\eps}
  -\int_\Omega p(x,t)h(x)\dd x\right|
  &\le \frac{1}{|\eps|}
  \int_{t-|\eps|M}^{t+|\eps|M}
  \norm{p(\cdot,\sigma)-p(\cdot,t)}_{L^1}\dd\sigma\\
  &\le 2M\sup_{|\sigma-t|\le |\eps|M}
  \norm{p(\cdot,\sigma)-p(\cdot,t)}_{L^1},
\end{align*}
which tends to zero. Applying the calculation to both endpoints proves the claim.
\end{proof}

The estimate shows why continuity with values in $L^1(\Omega)$ is sufficient: only the behavior of the power density in a shrinking time neighborhood of the terminal point is used. No pointwise continuity in $x$ is required. The bounded selector fields determine the local terminal direction, while compactly supported smooth selectors are used below when the resulting coefficients are identified as distributions.

For interior localization at a fixed time, the selectors are taken independently in
\begin{equation}\label{eq:local_selector_space}
  \delta t_c,\delta t_s\in C_c^\infty(\Omega).
\end{equation}
They are virtual test fields with the dimension of time, not additional physical time coordinates. If
$\int_\Omega(A_c\delta t_c+A_s\delta t_s)\dd x=0$ for all such pairs, then $A_c=A_s=0$ in $\Dcal'(\Omega)$. The weak balance below is tested directly against the full space $H^1(\Omega)$.

\section{Diffusion power-port decomposition}\label{sec:spatial}

Let $\Omega\subset\R^d$ be a bounded Lipschitz domain.  The chemical potential is denoted by $\mu$ and the diffusion flux by $\bm j$.  The smooth identity
\begin{equation}\label{eq:smooth_leibniz}
  \Div(\mu\bm j)
  =\mu\Div\bm j+\bm j\cdot\grad\mu
\end{equation}
is a local Leibniz rule.  Its interpretation as a channel decomposition requires a regularity class, a boundary convention, and independent terminal selectors.

\subsection{Sobolev fields and boundary power}

Recall
\begin{equation}\label{eq:Hdiv_definition}
  H(\Div;\Omega)
  :=\set{\bm j\in L^2(\Omega;\R^d):\Div\bm j\in L^2(\Omega)}.
\end{equation}
For $\bm j\in H(\Div;\Omega)$, its normal trace $\gamma_n\bm j\in H^{-1/2}(\partial\Omega)$ satisfies
\begin{equation}\label{eq:Hdiv_green}
  \int_\Omega \bm j\cdot\grad\eta\dd x
  +\int_\Omega (\Div\bm j)\eta\dd x
  =\dual{\gamma_n\bm j}{\gamma\eta}_{H^{-1/2},H^{1/2}}
\end{equation}
for every $\eta\in H^1(\Omega)$ \citep{GiraultRaviart1986,BrezziFortin1991}.

\begin{theorem}[Sobolev product rule and boundary port]\label{thm:sobolev_product}
Let $\mu\in H^1(\Omega)$ and $\bm j\in H(\Div;\Omega)$.  Then
$\mu\bm j$, $\mu\Div\bm j$, and $\bm j\cdot\grad\mu$ belong to $L^1(\Omega)$, and
\begin{equation}\label{eq:sobolev_product_rule}
  \Div(\mu\bm j)
  =\mu\Div\bm j+\bm j\cdot\grad\mu
  \qquad\text{in }\Dcal'(\Omega).
\end{equation}
Moreover,
\begin{equation}\label{eq:sobolev_boundary_port}
  \int_\Omega\mu\Div\bm j\dd x
  +\int_\Omega\bm j\cdot\grad\mu\dd x
  =\dual{\gamma_n\bm j}{\gamma\mu}_{H^{-1/2},H^{1/2}}.
\end{equation}
\end{theorem}

\begin{proof}
The integrability statements follow from H\"older's inequality.  For $\varphi\in C_c^\infty(\Omega)$, the function $\mu\varphi$ belongs to $H_0^1(\Omega)$. Hence
\[
  0=\int_\Omega (\Div\bm j)\mu\varphi\dd x
    +\int_\Omega\bm j\cdot\grad(\mu\varphi)\dd x,
\]
and expansion of the gradient gives \eqref{eq:sobolev_product_rule}.  Equation \eqref{eq:sobolev_boundary_port} is \eqref{eq:Hdiv_green} with $\eta=\mu$.
\end{proof}

Thus the local product rule and the boundary power identity are two parts of the same statement. In an open system the volume contributions do not cancel by themselves; their sum is exactly the normal chemical-power port. This is why the trace term is retained before imposing a closed-port condition.

Set
\begin{equation}\label{eq:sobolev_port_components}
  P_c:=\mu\Div\bm j,
  \qquad
  P_s:=\bm j\cdot\grad\mu.
\end{equation}
The subscripts refer to the conserved-variable and entropy channels.  Equation \eqref{eq:sobolev_boundary_port} shows that their spatial integrals sum to boundary chemical power.  They cancel globally only for a closed port, for example $\gamma_n\bm j=0$.

\subsection{Channelwise terminal variation}

Assume that $P_c,P_s\in C(I;L^1(\Omega))$ on an open time interval $I$. For measurable endpoint fields $\tau_c,\tau_s:\Omega\to I$, define
\begin{equation}\label{eq:split_port_extension}
\begin{aligned}
  \Gamma_{P}[\tau_c,\tau_s]
  :=&\int_\Omega\int_{t_0}^{\tau_c(x)}P_c(x,\sigma)\dd\sigma\dd x\\
    &+\int_\Omega\int_{t_0}^{\tau_s(x)}P_s(x,\sigma)\dd\sigma\dd x.
\end{aligned}
\end{equation}
On the synchronization diagonal,
\begin{equation}\label{eq:split_port_reconstruction}
  \Gamma_{P}[t,t]
  =\int_\Omega\int_{t_0}^{t}\Div(\mu\bm j)(x,\sigma)\dd\sigma\dd x.
\end{equation}

\begin{theorem}[Diffusion-port terminal rule]\label{thm:split_port_rule}
Fix $t\in I$, let $\delta t_c,\delta t_s\in L^\infty(\Omega)$, and set
\begin{equation}\label{eq:split_port_perturbation}
  G_\eps
  :=\Gamma_P[t+\eps\delta t_c,t+\eps\delta t_s].
\end{equation}
Then
\begin{equation}\label{eq:split_port_derivative}
\boxed{
  \left.\frac{\dd G_\eps}{\dd\eps}\right|_{\eps=0}
  =\int_\Omega
  \left(
    \mu\Div\bm j\,\delta t_c
    +\bm j\cdot\grad\mu\,\delta t_s
  \right)(x,t)\dd x.
}
\end{equation}
\end{theorem}

\begin{proof}
Apply \cref{thm:localized_spending} separately to the two component histories in \eqref{eq:split_port_extension}.
\end{proof}

At $\varepsilon=0$ the two endpoint fields coincide and \eqref{eq:split_port_reconstruction} recovers the original accumulated diffusion port. The selectors distinguish only the first-order channel directions. They do not introduce two physical histories, and the decomposition remains tied to the single local identity \eqref{eq:smooth_leibniz}.

The local notation used in applications is therefore
\begin{equation}\label{eq:diffusion_port_rule_citable}
\boxed{
  \delta\int_{t_0}^{t}\Div(\mu\bm j)\dd\tau
  =\mu\Div\bm j\,\delta t_c
   +\bm j\cdot\grad\mu\,\delta t_s.
}
\end{equation}
Equation \eqref{eq:diffusion_port_rule_citable} is the local shorthand for the tested identity \eqref{eq:split_port_derivative}.  It is not an ordinary derivative that assigns two unrelated values to one scalar upper limit.  The original history is recovered at $\eps=0$, while the two independent selectors appear only in the admissible perturbation family.

\subsection{Finite-energy fields}

For weak diffusion fields one often has only
\begin{equation}\label{eq:finite_energy_fields}
  \mu\in H^1(\Omega),
  \qquad
  \bm j\in L^2(\Omega;\R^d).
\end{equation}
Then $\mu\bm j\in L^1$ and $\bm j\cdot\grad\mu\in L^1$, but the ordinary product $\mu\Div\bm j$ need not be defined.

\begin{definition}[Finite-energy balance component]\label{def:diamond}
For \eqref{eq:finite_energy_fields}, define
\begin{equation}\label{eq:diamond_definition}
  \mu\diamond\Div\bm j
  :=\Div(\mu\bm j)-\bm j\cdot\grad\mu
  \qquad\text{in }\Dcal'(\Omega).
\end{equation}
Equivalently, for $\varphi\in C_c^\infty(\Omega)$,
\begin{equation}\label{eq:diamond_test}
\begin{aligned}
  \dual{\mu\diamond\Div\bm j}{\varphi}
  :=&-\int_\Omega\mu\bm j\cdot\grad\varphi\dd x\\
     &-\int_\Omega\varphi\,\bm j\cdot\grad\mu\dd x.
\end{aligned}
\end{equation}
\end{definition}

By construction,
\begin{equation}\label{eq:weak_leibniz}
  \Div(\mu\bm j)
  =\mu\diamond\Div\bm j+\bm j\cdot\grad\mu
  \qquad\text{in }\Dcal'(\Omega).
\end{equation}
If $\bm j\in H(\Div;\Omega)$, then
$\mu\diamond\Div\bm j=\mu\Div\bm j$ in distributions, so \eqref{eq:weak_leibniz} extends \eqref{eq:sobolev_product_rule}. The definition is canonical at the regularity \eqref{eq:finite_energy_fields}: both $\Div(\mu\bm j)$ and $\bm j\cdot\grad\mu$ are well-defined distributions, whereas multiplying the generally distributional quantity $\Div\bm j$ by an $H^1$ function is not available as an ordinary product.

For local selectors $\delta t_c,\delta t_s\in C_c^\infty(\Omega)$, the finite-energy channel differential is
\begin{equation}\label{eq:weak_channel_differential}
  \dual{\mu\diamond\Div\bm j}{\delta t_c}
  +\int_\Omega\bm j\cdot\grad\mu\,\delta t_s\dd x.
\end{equation}
This is the rigorous replacement of the right-hand side of \eqref{eq:diffusion_port_rule_citable} at finite energy.

\begin{proposition}[Strong--weak stability]\label{prop:diamond_stability}
Suppose
\begin{equation}\label{eq:diamond_convergence}
  \mu_n\to\mu\quad\text{strongly in }H^1(\Omega),
  \qquad
  \bm j_n\rightharpoonup\bm j\quad\text{weakly in }L^2(\Omega;\R^d).
\end{equation}
Then
\begin{equation}\label{eq:diamond_stability}
  \mu_n\diamond\Div\bm j_n
  \longrightarrow
  \mu\diamond\Div\bm j
  \qquad\text{in }\Dcal'(\Omega),
\end{equation}
and
$\bm j_n\cdot\grad\mu_n\to\bm j\cdot\grad\mu$ in distributions.
\end{proposition}

\begin{proof}
For a fixed smooth compactly supported test function, both terms in \eqref{eq:diamond_test} are products of one strongly convergent $L^2$ factor and one weakly convergent $L^2$ factor.  The same argument applies to the force--flux product.
\end{proof}

This stability is the form needed in compactness arguments for approximate diffusion fields: strong convergence of the potential transfers the weak flux limit through both channel components without requiring separate control of $\Div\bm j_n$.

For spatially constant selectors and a closed normal port, \eqref{eq:sobolev_boundary_port} gives
\begin{equation}\label{eq:closed_scalar_spending_relation}
  \delta\int_{t_0}^{t}\int_\Omega\Div(\mu\bm j)\dd x\dd\tau
  =D_{\rm d}(t)(\delta t_c-\delta t_s),
  \qquad
  D_{\rm d}(t):=-\int_\Omega\bm j\cdot\grad\mu\dd x.
\end{equation}
For arbitrary local selectors, the two terms remain separate local test pairings and do not reduce to one scalar boundary identity.

The product rule alone does not determine the sign of $D_{\rm d}$.  A constitutive law supplies that sign.  For example, if
\begin{equation}\label{eq:onsager_flux}
  \bm j=-\bm M\grad\mu,
  \qquad
  \bm M=\bm M^{\mathsf T}\succeq0,
\end{equation}
then the local diffusion power density is
\begin{equation}\label{eq:diffusion_power_density}
  -\bm j\cdot\grad\mu
  =\grad\mu\cdot\bm M\grad\mu\ge0.
\end{equation}
Thus channel decomposition and constitutive nonnegativity are logically distinct.

\section{Recovery of local species balance}\label{sec:balance}

Let
\begin{equation}\label{eq:species_residual_strong}
  R^{\rm d}:=\dot c+\Div\bm j
\end{equation}
denote the species residual for a regular diffusion field.  The conserved-variable selector in \eqref{eq:diffusion_port_rule_citable}, together with the terminal chemical power $\mu\dot c\,\delta t_c$, gives
\begin{equation}\label{eq:weighted_species_equation}
  \mu R^{\rm d}=0.
\end{equation}
The relevant issue is the behavior of the model on a space--time region where the chemical potential vanishes identically.  The following result gives a direct sufficient condition under which the weighted channel already contains the local balance.

\begin{theorem}[Constitutive recovery of local balance]\label{thm:constitutive_recovery}
Let $Q:=\Omega\times(t_0,T)$, let $c,\mu\in C^1(Q)$ and $\bm j\in C^1(Q;\R^d)$, and let $\zeta$ denote prescribed constitutive data. Define $R^{\rm d}$ by \eqref{eq:species_residual_strong}. Suppose that:
\begin{enumerate}[label=\textup{(\roman*)}]
  \item $\mu R^{\rm d}=0$ in $Q$;
  \item the flux law has the form
  \begin{equation}\label{eq:zero_force_zero_flux}
    \bm j=\mathcal J(x,t,c,\zeta,-\grad\mu),
    \qquad
    \mathcal J(x,t,c,\zeta,0)=0;
  \end{equation}
  \item for every open cylinder $U\times I\Subset\Int_Q\set{\mu=0}$ there are Banach spaces $\mathcal X_U$ and $\mathcal Z_U$, time-independent constitutive data $\zeta_U$ satisfying $\zeta|_{U\times I}=\zeta_U$, and a map
  \begin{equation}\label{eq:local_chemical_map}
    \mathcal M_U(\,\cdot\,;\zeta_U)
    \in C^1(\mathcal X_U;\mathcal Z_U)
  \end{equation}
  such that $c|_U\in C^1(I;\mathcal X_U)$,
  $\dot c(t)|_U\in T_{c(t)}\mathcal A_U$,
  $\mu|_U=\mathcal M_U(c|_U;\zeta_U)$ in $\mathcal Z_U$, and
  \begin{equation}\label{eq:isolated_zero_state}
    \Ker D_c\mathcal M_U(c(t)|_U;\zeta_U)
    \cap T_{c(t)}\mathcal A_U=\set{0}
    \qquad\text{for every }t\in I,
  \end{equation}
  where $T_{c(t)}\mathcal A_U$ is the admissible tangent space.
\end{enumerate}
Then
\begin{equation}\label{eq:local_species_recovered}
  \dot c+\Div\bm j=0
  \qquad\text{in }Q.
\end{equation}
\end{theorem}

\begin{proof}
At every point with $\mu\ne0$, \eqref{eq:weighted_species_equation} gives $R^{\rm d}=0$. If a point belongs to the boundary of $\set{\mu=0}$, it is approached by points with $\mu\ne0$; continuity of $R^{\rm d}$ therefore gives $R^{\rm d}=0$ there as well. It remains to consider an open cylinder $U\times I\Subset\Int_Q\set{\mu=0}$.

On this cylinder, $\mu=0$ and hence $\grad\mu=0$.  Assumption (ii) gives $\bm j=0$, so $\Div\bm j=0$ locally in distributions.  The identity
$\mathcal M_U(c(t)|_U;\zeta_U)=0$ can be differentiated in $\mathcal Z_U$ because $\zeta_U$ is fixed in time:
\[
  D_c\mathcal M_U(c(t)|_U;\zeta_U)
  [\dot c(t)|_U]=0.
\]
The injectivity condition \eqref{eq:isolated_zero_state} yields $\dot c=0$ on $U\times I$. Hence $R^{\rm d}=0$ throughout the interior of the zero set and therefore in all of $Q$.
\end{proof}

The proof uses two distinct mechanisms. Away from the zero-potential set, the weighted relation gives the residual directly. On a persistent zero-potential region, the conclusion instead follows from zero force, zero flux, and dynamical isolation of the corresponding equilibrium state. The latter condition is the substantive assumption: it excludes motion along a neutral family on which the chemical potential remains identically zero.

\begin{corollary}[Local energies with isolated critical states]\label{cor:local_energy_recovery}
Let $c,\mu\in C^1(Q)$ and $\bm j\in C^1(Q;\R^d)$. Assume \eqref{eq:weighted_species_equation}, a flux law $\bm j=\mathcal J(-\grad\mu)$ with $\mathcal J(0)=0$, and
\begin{equation}\label{eq:local_chemical_potential}
  \mu=f'(c)
\end{equation}
for a function $f\in C^1(\R)$ whose critical set
$\set{r\in\R:f'(r)=0}$ is discrete.  Then \eqref{eq:local_species_recovered} holds.
\end{corollary}

\begin{proof}
On each connected component of an open zero-potential region, continuity of $c$ and discreteness of the critical set force $c$ to take one constant critical value. Thus $\dot c=0$.  Equation \eqref{eq:zero_force_zero_flux} gives $\bm j=0$, and the local residual vanishes.  Outside such a region the proof of \cref{thm:constitutive_recovery} applies.
\end{proof}

For a gradient energy
\begin{equation}\label{eq:gradient_energy_zero_state}
  \mu=f'(c)-\kappa\Delta c,
\end{equation}
condition \eqref{eq:isolated_zero_state} is the absence of an admissible kernel of the linearized operator
$f''(c)-\kappa\Delta$ at a zero-potential state.  It holds at isolated coercive equilibria after the boundary and mass constraints are included.  Translation-invariant interfaces can carry neutral modes, in which case the weighted channel alone does not determine motion along the equilibrium family; actual Cahn--Hilliard front motion is governed by the balance-constrained transport dynamics \citep{Pego1989}.  Time-dependent external fields require the corresponding term $D_\zeta\mathcal M_U[\dot\zeta]$ in the differentiated equilibrium relation and are therefore outside the stationary-parameter hypothesis of \cref{thm:constitutive_recovery}.

\subsection{Weak mixed completion}

The preceding recovery theorem concerns regular fields and uses the constitutive structure of persistent zero-potential states.  At finite-energy regularity, local conservation is formulated directly in a dual space.  Use the Gelfand triple
\begin{equation}\label{eq:gelfand_triple}
  V=H^1(\Omega)\hookrightarrow H=L^2(\Omega)
  \cong H^*\hookrightarrow V^*,
  \qquad
  Y:=L^2(\Omega;\R^d).
\end{equation}
Let $G:V\to Y$ be $G\eta=\grad\eta$ and define the closed-port weak divergence
\begin{equation}\label{eq:B_operator}
  B:=-G^*:Y\to V^*,
  \qquad
  \dual{B\bm j}{\eta}_{V^*,V}
  :=-\int_\Omega\bm j\cdot\grad\eta\dd x.
\end{equation}
Then
\begin{equation}\label{eq:operator_power_identity}
  \dual{B\bm j}{\mu}_{V^*,V}
  +(\bm j,G\mu)_Y=0.
\end{equation}
For $v\in V^*$, define
\begin{equation}\label{eq:weak_residual}
  R:=v+B\bm j\in V^*.
\end{equation}

Let $\Psi:Y\to[0,+\infty]$ be a proper, convex, lower-semicontinuous flux power functional with $\Psi(0)=0$, and let $\mu\in V$ be the energetic effort. Define
\begin{equation}\label{eq:mixed_functional}
  \Lcal_\mu(v,\bm j,\lambda)
  :=\Psi(\bm j)+\dual{v}{\mu}_{V^*,V}
  -\dual{v+B\bm j}{\lambda}_{V^*,V},
  \qquad \lambda\in V.
\end{equation}
A triple $(v,\bm j,\lambda)$ is stationary when
\begin{subequations}\label{eq:mixed_stationarity}
\begin{align}
  \dual{w}{\mu-\lambda}_{V^*,V}&=0
  &&\text{for every }w\in V^*,\label{eq:mixed_stationarity_v}\\
  \dual{v+B\bm j}{\eta}_{V^*,V}&=0
  &&\text{for every }\eta\in V,\label{eq:mixed_stationarity_lambda}\\
  \Psi(\bm k)-\Psi(\bm j)
  +(\bm k-\bm j,G\lambda)_Y&\ge0
  &&\text{for every }\bm k\in Y.\label{eq:mixed_stationarity_j}
\end{align}
\end{subequations}

\begin{theorem}[Mixed stationarity]\label{thm:mixed_stationarity}
Conditions \eqref{eq:mixed_stationarity} are equivalent to
\begin{subequations}\label{eq:mixed_system}
\begin{align}
  \lambda&=\mu &&\text{in }V,\label{eq:mixed_lambda}\\
  v+B\bm j&=0 &&\text{in }V^*,\label{eq:mixed_balance}\\
  -G\mu&\in\partial\Psi(\bm j) &&\text{in }Y.\label{eq:mixed_flux}
\end{align}
\end{subequations}
\end{theorem}

\begin{proof}
The first two stationarity conditions give \eqref{eq:mixed_lambda} and \eqref{eq:mixed_balance} because the corresponding test spaces separate points.  The third condition is the subgradient inequality
$-G\lambda\in\partial\Psi(\bm j)$; use \eqref{eq:mixed_lambda}.  Every step is reversible.
\end{proof}

The three variations play different roles. The rate direction identifies the multiplier with the energetic effort, the multiplier direction imposes conservation against the full test space $V$, and the flux inequality supplies the constitutive contact. Keeping these directions separate is what yields an unweighted balance in $V^*$.

Testing \eqref{eq:mixed_balance} by $1\in V$ gives
\begin{equation}\label{eq:mass_conservation_weak}
  \dual{v}{1}_{V^*,V}=0,
\end{equation}
since $G1=0$.  Thus $v=\dot c$ conserves total mass.  In the regular model class of \cref{thm:constitutive_recovery}, the weighted selector already yields the same local balance.  The multiplier supplies its weak completion without requiring continuity of the residual or isolation of zero-potential states.

\section{Energy balance and process power}\label{sec:energy}

Let
\begin{equation}\label{eq:energy_trajectory_regularity}
  c\in W^{1,1}(0,T;V^*),
  \qquad
  \mu\in L^1(0,T;V),
\end{equation}
and let $\Ecal$ be a free energy such that $\Ecal(c(\cdot))$ is absolutely continuous and
\begin{equation}\label{eq:energy_chain_rule}
  \frac{\dd}{\dd t}\Ecal(c(t))
  =\dual{\dot c(t)}{\mu(t)}_{V^*,V}
  \qquad\text{for a.e. }t.
\end{equation}
This chain rule is model-dependent for nonconvex weak evolutions and is stated explicitly rather than inferred from formal differentiation \citep{Brezis2011,Mielke2023}.

Let $\Psi_t:Y\to[0,+\infty]$ be a measurable family of proper, convex, lower-semicontinuous flux power functions with $\Psi_t(0)=0$, and set
\begin{equation}\label{eq:force_and_process_power}
  \bm e(t):=-G\mu(t),
  \qquad
  D(t):=\Psi_t(\bm j(t))+\Psi_t^*(\bm e(t)).
\end{equation}
Assume that all terms below are integrable.

\begin{proposition}[Energy--power equivalence]\label{prop:energy_power_equivalence}
Suppose
\begin{equation}\label{eq:trajectory_balance}
  \dot c+B\bm j=0
  \qquad\text{in }V^*\text{ for a.e. }t.
\end{equation}
Then the constitutive contact
\begin{equation}\label{eq:constitutive_contact}
  \bm e(t)\in\partial\Psi_t(\bm j(t))
\end{equation}
is equivalent to
\begin{equation}\label{eq:energy_power_balance}
  \Ecal(c(t))+
  \int_s^tD(\tau)\dd\tau
  =\Ecal(c(s))
  \qquad(0\le s\le t\le T).
\end{equation}
Equivalently, under the same balance and chain-rule hypotheses, it suffices to assume the inequality in \eqref{eq:energy_power_balance} for $s=0$ and $t=T$; the Fenchel gap then vanishes almost everywhere, and equality follows on every subinterval.
\end{proposition}

\begin{proof}
By \eqref{eq:energy_chain_rule}, \eqref{eq:trajectory_balance}, and \eqref{eq:operator_power_identity},
\begin{equation}\label{eq:energy_rate_pairing}
  \frac{\dd}{\dd t}\Ecal(c(t))
  =-\dual{B\bm j}{\mu}
  =(\bm j,G\mu)_Y
  =-(\bm e,\bm j)_Y.
\end{equation}
Fenchel--Young gives the nonnegative gap
\begin{equation}\label{eq:fenchel_gap}
  \Psi_t(\bm j)+\Psi_t^*(\bm e)-(\bm e,\bm j)_Y\ge0,
\end{equation}
with equality exactly under \eqref{eq:constitutive_contact}.  Integration proves both directions.  This is the standard energy--power argument for generalized gradient systems \citep{MielkePeletierRenger2014,Mielke2023}.
\end{proof}

For an open normal port, assume $\bm j\in H(\Div;\Omega)$, let $g=\gamma_n\bm j$, and write the balance as $\dot c+\Div\bm j=0$.  The Green identity is
\begin{equation}\label{eq:open_operator_identity}
  \int_\Omega\mu\Div\bm j\dd x
  +\int_\Omega\bm j\cdot\grad\mu\dd x
  =\dual{g}{\gamma\mu}_{H^{-1/2},H^{1/2}}.
\end{equation}
Consequently,
\begin{equation}\label{eq:open_energy_balance}
  \Ecal(c(t))+
  \int_s^tD(\tau)\dd\tau
  +\int_s^t\dual{g(\tau)}{\gamma\mu(\tau)}\dd\tau
  =\Ecal(c(s)).
\end{equation}
With outward normal orientation, the last term is chemical power leaving the domain.  The operator $B$ in \eqref{eq:B_operator} is reserved for the closed-port realization.

Along a closed trajectory satisfying \eqref{eq:constitutive_contact}, Fenchel equality yields
\begin{equation}\label{eq:process_power_pairing}
  D(t)=(\bm e(t),\bm j(t))_Y
  =-\int_\Omega\bm j\cdot\grad\mu\dd x\ge0.
\end{equation}
Combining \eqref{eq:energy_rate_pairing} with the spending rule gives
\begin{equation}\label{eq:trajectory_spending_identity}
  D(t)\delta t_s-D(t)\delta t_c
  =D(t)\delta t_s
   +\dual{\dot c(t)\delta t_c}{\mu(t)}_{V^*,V}.
\end{equation}
If the local power density
\begin{equation}\label{eq:local_diffusion_power_density}
  p(x,t):=-\bm j(x,t)\cdot\grad\mu(x,t)
\end{equation}
belongs to $C(I;L^1(\Omega))$, \cref{thm:localized_spending} gives
\begin{equation}\label{eq:localized_trajectory_spending}
  \left.\frac{\dd}{\dd\eps}\right|_{\eps=0}
  \Gamma_p[t+\eps\delta t_s,t+\eps\delta t_c]
  =\int_\Omega p(x,t)(\delta t_s-\delta t_c)\dd x.
\end{equation}
The global identity \eqref{eq:trajectory_spending_identity} and the local test identity \eqref{eq:localized_trajectory_spending} are complementary statements: the first uses the closed-port duality pairing, while the second resolves the power density into independent local selector fields.

\section{A synchronized thermo-diffusion functional}\label{sec:single}

This section realizes the preceding structures within one scalar functional on a product space.  Storage variables, balance variables, and stored histories remain independent during differentiation and are synchronized only afterward.  The resulting stationarity is directional stationarity on a declared subspace, not unrestricted criticality with respect to every ambient copy.

\subsection{Storage, mixed, and history blocks}

Let $\Omega$ be a bounded $C^1$ domain and use closed normal ports in the derivation. Introduce an energetic concentration $c_E$, a kinematic concentration $c_K$, an energetic temperature $\theta_E$, and a mixed-block temperature $\theta_M$. Define
\begin{equation}\label{eq:single_free_energy}
  \mathcal F(c_E,\theta_E)
  :=\int_\Omega\psi(c_E,\grad c_E,\theta_E)\dd x.
\end{equation}
For smooth fields, write
\begin{equation}\label{eq:single_energy_effort}
  D_c\mathcal F(c_E,\theta_E)[\eta]
  =\int_\Omega\mu_{\mathcal F}\eta\dd x,
  \qquad
  \mu_{\mathcal F}
  =\partial_c\psi-\Div\partial_{\grad c}\psi,
\end{equation}
with the corresponding natural boundary condition.

Let $\mu_E$ be the effort in the energetic--kinematic bridge, let $\mu_M$ be the effort in the mixed block, and let $(\bar\mu,\bar{\bm j},\bar{\bm q})$ denote fixed stored-history copies. Let $s$ be entropy density, $v$ the concentration rate, and $\lambda$ an independent balance multiplier. For convex $C^1$ flux power densities $r_{\rm d}(x,\cdot)$ and $r_{\rm q}(x,\cdot)$, define
\begin{equation}\label{eq:single_mixed_block}
\begin{aligned}
  \mathfrak M_{\mu_M,\theta_M}(v,\bm j,\bm q,\lambda)
  :=\int_\Omega\bigl[
    &r_{\rm d}(x,\bm j)+r_{\rm q}(x,\bm q)
     +\mu_M v\\
    &-\lambda(v+\Div\bm j)
     +\bm q\cdot\grad\theta_M
  \bigr]\dd x.
\end{aligned}
\end{equation}
The bridge is
\begin{equation}\label{eq:single_bridge}
  \mathfrak B(c_E,c_K;\mu_E)
  :=\int_\Omega\mu_E(c_K-c_E)\dd x.
\end{equation}
It vanishes after synchronization but has nonzero derivatives in the two independent concentration directions.

For endpoint fields $\tau_c,\tau_s$, define the stored histories
\begin{equation}\label{eq:single_heat_history}
  \mathscr H_{\bar q}[\tau_s]
  :=\int_\Omega\int_{t_0}^{\tau_s(x)}
     \Div\bar{\bm q}(x,\sigma)\dd\sigma\dd x
\end{equation}
and
\begin{equation}\label{eq:single_diffusion_history}
\begin{aligned}
  \mathscr H_{\bar\mu\bar j}[\tau_c,\tau_s]
  :=&\int_\Omega\int_{t_0}^{\tau_c(x)}
       \bar\mu\Div\bar{\bm j}\dd\sigma\dd x\\
    &+\int_\Omega\int_{t_0}^{\tau_s(x)}
       \bar{\bm j}\cdot\grad\bar\mu\dd\sigma\dd x.
\end{aligned}
\end{equation}
At synchronized endpoints and history copies, \eqref{eq:single_diffusion_history} reconstructs the accumulated history of $\Div(\mu\bm j)$.

For a fixed positive scaling parameter $\tau_*$, set
\begin{equation}\label{eq:single_functional}
\boxed{
\begin{aligned}
  \widehat\Pi_t
  :=&\;\mathcal F(c_E,\theta_E)
    +\int_\Omega\theta_Es\dd x
    +\mathfrak B(c_E,c_K;\mu_E)\\
   &+\tau_*\mathfrak M_{\mu_M,\theta_M}(v,\bm j,\bm q,\lambda)
    +\mathscr H_{\bar q}[\tau_s]
    +\mathscr H_{\bar\mu\bar j}[\tau_c,\tau_s].
\end{aligned}
}
\end{equation}
The physical synchronization set at time $t$ is
\begin{equation}\label{eq:single_sync_set}
\begin{gathered}
  c_E=c_K=c,
  \qquad
  \theta_E=\theta_M=\theta,
  \qquad
  \mu_E=\mu_M=\bar\mu=\mu,\\
  \bar{\bm j}=\bm j,
  \qquad
  \bar{\bm q}=\bm q,
  \qquad
  v=\dot c,
  \qquad
  \tau_c=\tau_s=t.
\end{gathered}
\end{equation}
Denote this set by $\mathscr S_t$ and set $\Pi_t=\widehat\Pi_t|_{\mathscr S_t}$ after differentiation.

\subsection{Admissible directions and first variation}

At a point of $\mathscr S_t$, take
\begin{equation}\label{eq:single_direction_tuple}
  \bm h=(\eta,\vartheta,w,\bm k_{\rm d},\bm k_{\rm q},\ell,
  \delta t_c,\delta t_s),
\end{equation}
where the scalar directions belong to $C_c^\infty(\Omega)$ and the flux directions to $C_c^\infty(\Omega;\R^d)$ in the bulk calculation. The nonzero ambient components are
\begin{equation}\label{eq:single_direction_components}
\begin{gathered}
  \delta c_E=\eta,
  \qquad
  \delta c_K=v\,\delta t_c,
  \qquad
  \delta\theta_E=\vartheta,
  \qquad
  \delta s=\dot s\,\delta t_s,\\
  \delta v=w,
  \qquad
  \delta\bm j=\bm k_{\rm d},
  \qquad
  \delta\bm q=\bm k_{\rm q},
  \qquad
  \delta\lambda=\ell,\\
  \delta\tau_c=\delta t_c,
  \qquad
  \delta\tau_s=\delta t_s.
\end{gathered}
\end{equation}
The copies $\mu_E,\mu_M,\theta_M$ and the stored histories $(\bar\mu,\bar{\bm j},\bar{\bm q})$ remain fixed in this direction. Define
\begin{equation}\label{eq:single_directional_derivative}
  \delta\Pi_t[\bm h]
  :=D\widehat\Pi_t\big|_{\mathscr S_t}[\bm h].
\end{equation}
This is the ordinary directional derivative of one scalar functional, restricted to the declared admissible direction space.

The vertical state variation $\eta$ and the trajectory-generated direction $v\,\delta t_c$ are independent before synchronization.  The latter is one selected history direction and does not represent every concentration variation; no division by $v$ is used.  Likewise, $\delta t_c$ and $\delta t_s$ are independent local selector fields.  This separation is the reason for the energetic and kinematic concentration copies.

\begin{theorem}[First variation of the synchronized functional]\label{thm:single_first_variation}
Assume that $\psi$ is continuously differentiable in its displayed variables, that $r_{\rm d}$ and $r_{\rm q}$ are continuously differentiable in their flux variables, and that the synchronized fields, histories, and directions have the regularity required by the integrations and terminal derivatives above.  In the closed-port calculation, the normal components of the flux variations vanish. Then
\begin{equation}\label{eq:single_first_variation}
\boxed{
\begin{aligned}
  \delta\Pi_t[\bm h]
  =&\int_\Omega\Bigl[
      (\mu_{\mathcal F}-\mu)\eta
      +(\partial_\theta\psi+s)\vartheta
      +\mu(v+\Div\bm j)\delta t_c\\
   &\hspace{27mm}
      +(\theta\dot s+\Div\bm q
        +\bm j\cdot\grad\mu)\delta t_s
    \Bigr]\dd x\\
   &+\tau_*\int_\Omega\Bigl[
      (\mu-\lambda)w
      +(\partial_{\bm j}r_{\rm d}+\grad\lambda)\cdot\bm k_{\rm d}\\
   &\hspace{25mm}
      +(\partial_{\bm q}r_{\rm q}+\grad\theta)\cdot\bm k_{\rm q}
      -(v+\Div\bm j)\ell
    \Bigr]\dd x.
\end{aligned}
}
\end{equation}
\end{theorem}

\begin{proof}
The storage and bridge blocks contribute
\[
  \int_\Omega\left[
    \mu_{\mathcal F}\eta
    +(\partial_\theta\psi+s)\vartheta
    +\theta\dot s\,\delta t_s
    +\mu(v\delta t_c-\eta)
  \right]\dd x.
\]
The mixed block gives the second integral in \eqref{eq:single_first_variation}; the closed-port integration by parts changes
$-\int_\Omega\lambda\Div\bm k_{\rm d}\dd x$
into $\int_\Omega\grad\lambda\cdot\bm k_{\rm d}\dd x$.  The terminal derivatives of \eqref{eq:single_heat_history} and \eqref{eq:single_diffusion_history} are
\[
  \int_\Omega\Div\bm q\,\delta t_s\dd x,
  \qquad
  \int_\Omega
  \left(\mu\Div\bm j\,\delta t_c
       +\bm j\cdot\grad\mu\,\delta t_s\right)\dd x,
\]
by \cref{thm:split_port_rule}. Adding the contributions proves the formula.
\end{proof}

Formula \eqref{eq:single_first_variation} displays the contribution of each block before the copies are identified. The storage and bridge terms generate conjugacy and the trajectory-based species power, the mixed block provides the independently tested balance and flux laws, and the locked histories supply the two terminal channels. Synchronization is therefore imposed after differentiation, not used to insert the field equations into the functional.

\begin{theorem}[Generated thermo-diffusion structure]\label{thm:single_generation}
The directional stationarity condition
\begin{equation}\label{eq:single_stationarity}
  \delta\Pi_t[\bm h]=0
  \qquad\text{for every admissible }\bm h
\end{equation}
is equivalent to
\begin{subequations}\label{eq:single_generated_system}
\begin{align}
  \mu&=\frac{\delta\mathcal F}{\delta c}
  =\partial_c\psi-\Div\partial_{\grad c}\psi,
  &&\text{energetic conjugacy},\label{eq:single_mu}\\
  s&=-\partial_\theta\psi,
  &&\text{thermal conjugacy},\label{eq:single_s}\\
  v+\Div\bm j&=0,
  &&\text{species balance},\label{eq:single_balance}\\
  \lambda&=\mu,
  &&\text{balance effort},\label{eq:single_lambda}\\
  \partial_{\bm j}r_{\rm d}(x,\bm j)+\grad\mu&=0,
  &&\text{diffusion law},\label{eq:single_diffusion_law}\\
  \partial_{\bm q}r_{\rm q}(x,\bm q)+\grad\theta&=0,
  &&\text{heat-flux law},\label{eq:single_heat_law}\\
  \theta\dot s+\Div\bm q+\bm j\cdot\grad\mu&=0,
  &&\text{entropy channel}.\label{eq:single_entropy}
\end{align}
\end{subequations}
The coefficient $\mu(v+\Div\bm j)$ then vanishes by the independently generated balance; it is not divided by $\mu$.
\end{theorem}

\begin{proof}
Independence of $\eta,\vartheta,w,\bm k_{\rm d},\bm k_{\rm q},\ell$, and $\delta t_s$ in \eqref{eq:single_first_variation} gives \eqref{eq:single_generated_system}.  The $\delta t_c$ coefficient is the weighted channel and follows from \eqref{eq:single_balance}.  Conversely, substitution of \eqref{eq:single_generated_system} into \eqref{eq:single_first_variation} makes every coefficient vanish.
\end{proof}

Under the hypotheses of \cref{thm:constitutive_recovery}, the weighted selector coefficient
$\mu(v+\Div\bm j)$ in \eqref{eq:single_first_variation} already implies the local species balance for regular fields.  The independent multiplier direction $\ell$ is retained because it yields the same balance in $V^*$ without the strong regularity and isolated-zero-state assumptions.  It is therefore a weak completion of the conserved-variable channel, not a second conservation law.

For scalar endpoints in a closed port, \eqref{eq:sobolev_boundary_port} gives the exact history identity
\begin{equation}\label{eq:single_history_equivalence}
  \mathscr H_{\mu j}[t_c,t_s]
  =H_{D_{\rm d}}(t_c)-H_{D_{\rm d}}(t_s)
  =-\Gamma_{D_{\rm d}}(t_s,t_c).
\end{equation}
Thus the split-port history and the scalar spending history are two orientations of one stored process power.

The two directly reusable terminal formulas are present in the same functional:
\begin{equation}\label{eq:single_diffusion_rule}
\boxed{
  \delta\int_{t_0}^{t}\Div(\mu\bm j)\dd\tau
  =\mu\Div\bm j\,\delta t_c
   +\bm j\cdot\grad\mu\,\delta t_s,
}
\end{equation}
and, for a closed diffusion port with $D_{\rm d}$ from \eqref{eq:closed_scalar_spending_relation},
\begin{equation}\label{eq:single_spending_rule}
\boxed{
  \delta\int_{t_0}^{t}D_{\rm d}(\tau)\dd\tau
  =D_{\rm d}(t)\delta t_s-D_{\rm d}(t)\delta t_c.
}
\end{equation}
The histories in \eqref{eq:single_diffusion_rule} and \eqref{eq:single_spending_rule} describe the same closed process with opposite orientation; they are not two independently stored powers.

For the quadratic flux power densities
\begin{equation}\label{eq:single_quadratic_flux_powers}
  r_{\rm d}(x,\bm j)=\frac12\bm j\cdot\bm M^{-1}\bm j,
  \qquad
  r_{\rm q}(x,\bm q)=\frac12\bm q\cdot\bm K^{-1}\bm q,
\end{equation}
where $\bm M$ and $\bm K$ are symmetric positive definite, \eqref{eq:single_diffusion_law}--\eqref{eq:single_heat_law} give
\begin{equation}\label{eq:single_onsager_closures}
  \bm j=-\bm M\grad\mu,
  \qquad
  \bm q=-\bm K\grad\theta.
\end{equation}
Hence
\begin{equation}\label{eq:single_entropy_production}
  \theta\dot s+\Div\bm q
  =\grad\mu\cdot\bm M\grad\mu\ge0.
\end{equation}
For $\theta>0$, combination with Fourier conduction gives the usual nonnegative entropy-production density after rewriting the heat-flux divergence in terms of $\bm q/\theta$.

At finite energy, the conserved-variable channel in \eqref{eq:single_first_variation} is replaced by
\begin{equation}\label{eq:single_weak_channel}
  \dual{v}{\mu\delta t_c}_{V^*,V}
  +\dual{\mu\diamond\Div\bm j}{\delta t_c},
\end{equation}
and the entropy-side diffusion term remains
\begin{equation}\label{eq:single_weak_entropy_term}
  \int_\Omega\bm j\cdot\grad\mu\,\delta t_s\dd x.
\end{equation}
Here $\delta t_c$ is chosen from a multiplier class for which $\mu\delta t_c\in V$.  The independent multiplier gives $v+B\bm j=0$ in $V^*$.  For regular fields satisfying \cref{thm:constitutive_recovery}, the weighted selector gives the same balance directly.  Thus the augmented functional contains the physical weighted channel together with its weak completion.

\section{Cahn--Hilliard realization}\label{sec:ch}

Let $d\le3$, let $\kappa>0$, and let $F\in C^1(\R)$ be bounded below. Assume
\begin{equation}\label{eq:CH_growth}
  \abs{F'(r)}\le C_F(1+\abs r^q),
  \qquad 1\le q\le5,
\end{equation}
so that the weak nonlinear term is defined for $H^1(\Omega)$ fields in three dimensions. Define
\begin{equation}\label{eq:CH_energy}
  \Ecal_{\rm CH}(c)
  :=\int_\Omega\left(F(c)+\frac{\kappa}{2}\abs{\grad c}^2\right)\dd x.
\end{equation}
A weak chemical potential $\mu\in H^1(\Omega)$ satisfies
\begin{equation}\label{eq:CH_chemical_weak}
  \int_\Omega\mu\zeta\dd x
  =\int_\Omega F'(c)\zeta\dd x
   +\kappa\int_\Omega\grad c\cdot\grad\zeta\dd x
  \qquad\text{for every }\zeta\in H^1(\Omega).
\end{equation}
For smooth fields, this is
$\mu=F'(c)-\kappa\Delta c$ with $\partial_n c=0$.

Assume
\begin{equation}\label{eq:CH_trajectory_regularity}
  c\in W^{1,1}(0,T;H^1(\Omega)^*)
      \cap L^\infty(0,T;H^1(\Omega)),
  \qquad
  \mu\in L^2(0,T;H^1(\Omega)),
\end{equation}
and the chain rule
\begin{equation}\label{eq:CH_chain_rule}
  \frac{\dd}{\dd t}\Ecal_{\rm CH}(c(t))
  =\dual{\dot c(t)}{\mu(t)}_{H^1(\Omega)^*,H^1(\Omega)}
  \qquad\text{for a.e. }t.
\end{equation}
These hypotheses isolate the energy regularity required below; classical weak-solution results provide concrete sufficient assumptions, while energetic formulations with dynamic boundary effects illustrate the additional trace terms required for open interfaces \citep{ElliottZheng1986,ElliottGarcke1996,LiuWu2019}.

Let $\bm M\in L^\infty(\Omega;\R^{d\times d})$ be symmetric and uniformly elliptic:
\begin{equation}\label{eq:CH_mobility}
  m_0\abs{\xi}^2
  \le \xi\cdot\bm M(x)\xi
  \le m_1\abs{\xi}^2
\end{equation}
for some $0<m_0\le m_1$ and almost every $x$.  Set
\begin{equation}\label{eq:CH_flux_power}
  \Psi_{\rm CH}(\bm j)
  :=\frac12\int_\Omega\bm j\cdot\bm M^{-1}\bm j\dd x
\end{equation}
and
\begin{equation}\label{eq:CH_mixed_functional}
  \Lcal_{\rm CH}(v,\bm j,\lambda)
  :=\Psi_{\rm CH}(\bm j)
   +\dual{v}{\mu}
   -\dual{v+B\bm j}{\lambda}.
\end{equation}

\begin{theorem}[Cahn--Hilliard realization]\label{thm:CH_realization}
For a rate $v\in H^1(\Omega)^*$, the following are equivalent:
\begin{enumerate}[label=\textup{(\roman*)}]
  \item there exist $\bm j\in L^2(\Omega;\R^d)$ and $\lambda\in H^1(\Omega)$ stationary for \eqref{eq:CH_mixed_functional};
  \item
  \begin{subequations}\label{eq:CH_mixed_system}
  \begin{align}
    v+B\bm j&=0 &&\text{in }H^1(\Omega)^*,\label{eq:CH_balance}\\
    \bm j&=-\bm M\grad\mu &&\text{in }L^2(\Omega;\R^d),\label{eq:CH_flux}\\
    \lambda&=\mu &&\text{in }H^1(\Omega);\label{eq:CH_lambda}
  \end{align}
  \end{subequations}
  \item
  \begin{equation}\label{eq:CH_weak_rate}
    \dual{v}{\eta}
    +\int_\Omega\bm M\grad\mu\cdot\grad\eta\dd x=0
    \qquad\text{for every }\eta\in H^1(\Omega).
  \end{equation}
\end{enumerate}
For smooth fields and $v=\dot c$, this gives
\begin{equation}\label{eq:CH_strong_system}
  \dot c=\Div(\bm M\grad\mu),
  \qquad
  \mu=F'(c)-\kappa\Delta c,
\end{equation}
with $\bm M\grad\mu\cdot\bm n=0$ and $\partial_n c=0$.
\end{theorem}

\begin{proof}
The equivalence of (i) and (ii) follows from \cref{thm:mixed_stationarity}.  Substitution of \eqref{eq:CH_flux} into \eqref{eq:CH_balance} gives \eqref{eq:CH_weak_rate}; the converse is obtained by defining $\bm j$ through \eqref{eq:CH_flux} and setting $\lambda=\mu$.  The strong form follows whenever the displayed derivatives and traces exist.
\end{proof}

\begin{corollary}[Mass, energy, and terminal channels]\label{cor:CH_structure}
Let $v=\dot c$ and assume \eqref{eq:CH_trajectory_regularity}--\eqref{eq:CH_chain_rule} and the equivalent conditions of \cref{thm:CH_realization}. Then
\begin{equation}\label{eq:CH_mass_conservation}
  \dual{\dot c(t)}{1}=0
  \qquad\text{for a.e. }t,
\end{equation}
and
\begin{equation}\label{eq:CH_energy_balance}
  \Ecal_{\rm CH}(c(t))
  +\int_s^t\int_\Omega
     \grad\mu\cdot\bm M\grad\mu\dd x\dd\tau
  =\Ecal_{\rm CH}(c(s)).
\end{equation}
At almost every time,
\begin{equation}\label{eq:CH_weak_port}
  \Div(\mu\bm j)
  =\mu\diamond\Div\bm j+\bm j\cdot\grad\mu
  \qquad\text{in }\Dcal'(\Omega),
\end{equation}
with
\begin{equation}\label{eq:CH_power_density}
  -\bm j\cdot\grad\mu
  =\grad\mu\cdot\bm M\grad\mu\ge0.
\end{equation}
Consequently, the local diffusion-port rule is \eqref{eq:weak_channel_differential}, and the closed scalar spending rule is
\begin{equation}\label{eq:CH_spending_rule}
  \delta\int_{t_0}^{t}D_{\rm CH}(\tau)\dd\tau
  =D_{\rm CH}(t)(\delta t_s-\delta t_c),
  \qquad
  D_{\rm CH}(t):=\int_\Omega\grad\mu\cdot\bm M\grad\mu\dd x.
\end{equation}
\end{corollary}

\begin{proof}
Test \eqref{eq:CH_weak_rate} by $1$ to obtain \eqref{eq:CH_mass_conservation}, and by $\mu$ to obtain the differential form of \eqref{eq:CH_energy_balance}.  Equations \eqref{eq:CH_weak_port}--\eqref{eq:CH_power_density} follow from \cref{def:diamond} and \eqref{eq:CH_flux}.  The terminal statement is \cref{thm:calibrated_spending} applied to $D_{\rm CH}$.
\end{proof}

For a strong Cahn--Hilliard trajectory, the weighted selector alone recovers the local balance whenever an open zero-potential state is isolated in the admissible mass and boundary class, as specified in \cref{thm:constitutive_recovery}.  The mixed formulation \eqref{eq:CH_mixed_functional} remains valid without this isolation condition and is the appropriate finite-energy statement.

The same equations arise from the isothermal restriction of \eqref{eq:single_functional}: set $\theta=\theta_0$, $\bm q=0$, $\mathcal F=\Ecal_{\rm CH}$, and use the quadratic flux power density in \eqref{eq:CH_flux_power}.  The storage direction gives \eqref{eq:CH_chemical_weak}, the mixed directions give \eqref{eq:CH_mixed_system}, and the history selectors give \eqref{eq:single_diffusion_rule} and \eqref{eq:CH_spending_rule}.  This adds calibrated terminal routing and a finite-energy port interpretation to the standard mixed Cahn--Hilliard structure \citep{Gurtin1996,MieheHildebrandBoger2014}.

For the quartic double well
\begin{equation}\label{eq:quartic_energy}
  F(c)=\frac14(c^2-1)^2,
\end{equation}
relation \eqref{eq:CH_chemical_weak} becomes
\begin{equation}\label{eq:quartic_chemical_potential}
  \int_\Omega\mu\zeta\dd x
  =\int_\Omega c(c^2-1)\zeta\dd x
  +\kappa\int_\Omega\grad c\cdot\grad\zeta\dd x.
\end{equation}
For constant isotropic mobility, elimination of $\mu$ gives the classical fourth-order Cahn--Hilliard equation \citep{CahnHilliard1958,GurtinPolignoneVinals1996}.

\section{Discussion and conclusion}\label{sec:conclusion}

The analysis identifies the mathematical content of the two channel rules used in path-dependent entropic Lagrangians.  Endpoint calibration and cocycle additivity determine the oriented spending differential, while local selector fields provide the test-function localization required for field equations.  The diffusion port has an $H(\Div)$ form with normal-trace boundary power and a finite-energy form based on $\mu\diamond\Div\bm j$.

For regular diffusion models, the weighted species channel yields local conservation when every persistent zero-potential state is dynamically isolated in the admissible state class.  Local energies with discrete critical states satisfy this condition directly; gradient energies satisfy it at isolated coercive equilibria after the boundary and mass constraints are included.  The independent multiplier extends the balance to the natural weak space when these strong-solution hypotheses are unavailable.

The synchronized thermo-diffusion functional places storage conjugacy, local conservation, flux laws, entropy production, and terminal routing in one directional stationarity structure.  The Cahn--Hilliard specialization verifies the construction at finite-energy regularity.  Subsequent applications can invoke the spending rule, the diffusion-port rule, and the constitutive recovery theorem after specifying the selector class, boundary port, and the structure of zero-potential states.

\appendix

\section{Auxiliary multichannel extension}\label{app:multichannel}

Let $D_e\in L^1(t_0,T)$ have primitives $H_e$, and let $A\in\R^{m\times q}$ satisfy
\begin{equation}\label{eq:incidence_zero_sum}
  \bm 1^{\mathsf T}A=0.
\end{equation}
Define
\begin{equation}\label{eq:incidence_extension}
  \Gamma_A(\tau_1,\ldots,\tau_m)
  :=\sum_{e=1}^q\sum_{\alpha=1}^m
  A_{\alpha e}H_e(\tau_\alpha).
\end{equation}
At a common Lebesgue point $t$,
\begin{equation}\label{eq:incidence_differential}
  D\Gamma_A(t\bm1)[\bm\xi]
  =(A\bm D(t))\cdot\bm\xi,
  \qquad
  D\Gamma_A(t\bm1)[\xi\bm1]=0.
\end{equation}
Thus each column assigns one process power to a zero-sum family of terminal channels.  This is the direct multichannel extension of \eqref{eq:spending_rule_citable}; graph incidence matrices provide a standard realization \citep{vanDerSchaftMaschke2002}.

\section{Bounded-variation extension}\label{app:bvdm}

For $\mu\in BV(\Omega)\cap L^\infty(\Omega)$ and
$\bm j\in DM^\infty(\Omega)$, the Anzellotti pairing $(\bm j,D\mu)$ is the Radon measure defined distributionally by
\begin{equation}\label{eq:anzellotti_pairing}
  \dual{(\bm j,D\mu)}{\varphi}
  :=-\int_\Omega\mu^*\varphi\dd(\Div\bm j)
    -\int_\Omega\mu\bm j\cdot\grad\varphi\dd x,
  \qquad \varphi\in C_c^1(\Omega).
\end{equation}
It satisfies
\begin{equation}\label{eq:anzellotti_product}
  \Div(\mu\bm j)
  =\mu^*\Div\bm j+(\bm j,D\mu),
  \qquad
  \abs{(\bm j,D\mu)}
  \le\norm{\bm j}_{L^\infty}\abs{D\mu}.
\end{equation}
The absolutely continuous part is $\bm j\cdot\grad\mu\,\mathcal L^d$, while the singular part retains interfacial power on jump and Cantor sets \citep{Anzellotti1983,ChenFrid1999,ChenComiTorres2019,ComiCrastaDeCiccoMalusa2024}.  This extension is auxiliary to the $H^1\times L^2$ theory used in the main text; the underlying $BV$ notation and fine properties are as in \citet{AmbrosioFuscoPallara2000}.

\section*{Data availability}
No data were generated or analyzed in this theoretical study.

\section*{Declaration of competing interest}
The author declares no competing interest.

\printbibliography

\end{document}